\documentclass[12pt]{iopart}
\usepackage{amssymb}
\usepackage{iopams}
\usepackage{amsthm}
\usepackage{amscd}
\usepackage{amsfonts}
\usepackage{amsbsy}
\usepackage{latexsym}
\usepackage{graphicx}

\catcode`\@=11
\renewcommand\footnoterule{%
  \kern-3\p@
  \hrule\@width.4\columnwidth
  \kern2.6\p@}
\renewcommand\@makefntext[1]{%
    \parindent 1em\noindent
    \hb@xt@1.8em{\hss$^{\@thefnmark}$)}\hspace{2pt}%
    \footnotesize\rmfamily#1}  
\def\@makefnmark{\hspace{.5pt}\hbox{$^{\@thefnmark}$%
\hspace{-1pt})}} 
\catcode`\@=11
\renewcommand\footnoterule{%
  \kern-3\p@
  \hrule\@width.4\columnwidth
  \kern2.6\p@}
\renewcommand\@makefntext[1]{%
    \parindent 1em\noindent
    \hb@xt@1.8em{\hss$^{\@thefnmark}$)}\hspace{2pt}%
    \footnotesize\rmfamily#1}  
\def\@makefnmark{\hspace{.5pt}\hbox{$^{\@thefnmark}$%
\hspace{-1pt}}} 

 \newtheorem{thm}{Theorem}[section]
 \newtheorem{cor}[thm]{Corollary}
 \newtheorem{lem}[thm]{Lemma}
 \newtheorem{prop}[thm]{Proposition}
 \theoremstyle{definition}
 \newtheorem{defn}[thm]{Definition}
 \theoremstyle{remark}
 \newtheorem{rem}[thm]{Remark}

\newcommand{\Lc}{{\mathcal L}}
\newcommand{\Hil}{\mathcal H}

\newcommand{\cH}{\mathcal{H}}

\newcommand{\cC}{\mathcal{C}}
\newcommand{\eqref}[1]{(\ref{#1})}

\begin{document}
\title[Generalized Riesz systems and orthonormal sequences]{Generalized Riesz systems and orthonormal sequences in Krein spaces}

\author{Fabio Bagarello${}^a$\  and \ Sergiusz Ku\.{z}el${}^b$}

\address{${}^a$\
DEIM -Dipartimento di Energia, ingegneria dell' Informazione e modelli Matematici,
Scuola Politecnica, Universit\`a di Palermo, I-90128  Palermo, Italy
and  INFN, Sezione di Napoli.  \\
${}^b$\  AGH University of Science and Technology, Faculty of Applied Mathematics, 30-059 al. Mickiewicza 30, Krak\'{o}w, Poland}
\eads{\mailto{fabio.bagarello@unipa.it},\ 
\mailto{kuzhel@agh.edu.pl}}
\begin{abstract}
We analyze special classes of bi-orthogonal sets of vectors in Hilbert and in Krein spaces, 
and their relations with generalized Riesz systems. In this way, the notion of the first/second type sequences
is introduced and studied.   
We also discuss their relevance in some concrete quantum mechanical system driven by manifestly non self-adjoint Hamiltonians. 
\end{abstract}



%
%
%
%
%
%
%
%
%


\section{Introduction}
The employing of non self-adjoint operators for the description of experimentally observable data
goes back to the early days of quantum mechanics. In the past twenty years, the steady interest in this subject grew 
considerably after it has been discovered  \cite{B1, D2}  that the spectrum of the 
manifestly non self-adjoint Hamiltonian
\begin{equation}\label{e1b}
    H=-\frac{d^2}{dx^2} + x^2(ix)^\epsilon, \qquad  0<\epsilon<2
\end{equation}
 is real. It was conjectured \cite{B1} that the reality of eigenvalues of $H$ is a consequence of its
$\mathcal{P}\mathcal{T}$-symmetry: ${\mathcal P}{\mathcal
T}H=H{\mathcal P}{\mathcal T}$,  where the space 
parity operator $\mathcal{P}$ and the complex conjugation operator
$\mathcal{T}$ are defined as follows:
$({\mathcal P}f)(x)=f(-x)$ and $({\mathcal T}f)(x)=\overline{f(x)}.$
This gave rise to a consistent development of theory of $\mathcal{PT}$-symmetric Hamiltonians in Quantum Physics,
see  \cite{bagbook_thebook,  bagpsasstra, BE} and  references therein.

Usually, $\mathcal{PT}$-symmetric Hamiltonians can be
interpreted as self-adjoint ones for a suitable choice of \emph{indefinite} inner product.
For instance, the operator $H$ in (\ref{e1b}) is  self-adjoint with respect to the indefinite inner product
$[f,g]=\int_{-\infty}^{\infty}f(-x)\overline{g(x)}dx$ in the Krein space  $(\Lc^2(\Bbb R), [\cdot,\cdot])$ 
(see  Subsection \ref{sec4a} for the definition of Krein spaces).  The eigenstates of $H$ lose the property of being Riesz basis in the original Hilbert space 
$\Lc^2(\Bbb R)$ but they still form a complete set in $\Lc^2(\Bbb R)$ \cite{new1, new2}.
 Moreover, they form a sequence which is orthonormal with respect to the indefinite
inner product $[\cdot, \cdot]$. 
Such kind of phenomenon is typical for  $\mathcal{PT}$-symmetric Hamiltonians and it gives rise to a natural problem:
\emph{What can we say about properties of vectors which form a complete set in a Hilbert space and, 
additionally,   are orthonormal with respect
to the indefinite inner product?}

The main objective of the paper is to investigate  this problem with the use of theory of generalized Riesz systems (GRS) and 
$\mathcal{G}$-quasi bases. These two concepts were originally introduced in \cite{Inoue} and \cite{bag2013JMP}, respectively 
and then analyzed in a series of papers, see, e.g. \cite[Chapter 3]{bagbook_thebook}, \cite{BB, BIT},  \cite{Inoue1}. 
The motivation for introducing GRS and $\mathcal{G}$-quasi bases was the need to put on a mathematically rigorous 
ground several physical models where the eigenstates of some non self-adjoint operator 
were usually claimed to be bases, while they were not.  

The paper is structured as follows.  Section \ref{sec2} contains facts related to GRS 
and $\mathcal{G}$-quasi bases.  We slightly change the original definition of GRS given in \cite{Inoue}  
putting in evidence the role of a self-adjoint operator $Q$, 
since it is more convenient for our purpose:
\emph{a sequence $\{\phi_n\}$ in a Hilbert space $\cH$ is called a generalized Riesz system (GRS) if there exists a self-adjoint operator $Q$ in $\cH$ and 
an orthonormal basis (ONB) $\{e_n\}$  such that  $e_n\in{D}(e^{Q/2})\cap{D}(e^{-Q/2})$ and $\phi_n=e^{Q/2}e_n$.}

A GRS $\{\phi_n\}$ is a Riesz basis if and only if $Q$ is a bounded operator.  

For each GRS  $\{\phi_n\}$, the dual  GRS is determined by the formula $\{\psi_n=e^{-Q/2}e_n\}$.
The dual GRS $\{\phi_n\}$ and  $\{\psi_n\}$ are biorthogonal. 

 In Subsection \ref{subsection2.2}, the phenomenon of nonuniqueness of a self-adjoint operator $Q$ in the formulas
 \begin{equation}\label{intr1}
 \phi_n=e^{Q/2}e_n, \qquad   \psi_n=e^{-Q/2}e_n
 \end{equation}
 is investigated in the case of regular biorthogonal sequences $\{\phi_n\}$ and  $\{\psi_n\}$. 
We firstly prove that the operator $Q$ and ONB $\{e_n\}$ are determined uniquely for dual bases.
In general, this property does not hold for regular biorthogonal sequences. 
To describe all possible self-adjoint operators $Q$ in \eqref{intr1}  we consider
a positive densely defined operator $G_0$ which acts as 
$G_0\phi_n=\psi_n$ and then is extended on $D(G_0)=span\{\phi_n\}$  by the linearity. 
The proved statement is: \emph{for given $\{\phi_n\}$ and  $\{\psi_n\}$, the set of admissible operators $Q$ in
\eqref{intr1} is in one-to-one correspondence with the set of extremal
extensions $G$ of $G_0$, precisely, $Q=-\ln{G}$.}   This fact allows one to characterize the important
case where $Q$ is determined uniquely.   

Section \ref{sec4} contains the main results. We begin with the simple fact that each complete sequence
 $\{\phi_n\}$ which is orthonormal in a Krein space $(\cH, [\cdot,\cdot])$  (briefly,  $J$-orthonormal) 
 is a GRS and therefore, it can be expressed by \eqref{intr1} for some choice of $Q$. 
If $Q$  is uniquely determined, then the  anti-commutation relation  $JQ=-QJ$ holds. 
For this reason, the anti-commutation relation is reasonable to keep in the general case
(if $Q$ is not determined uniquely, then we cannot state, in general, that $JQ=-QJ$): 
we say that a complete $J$-orthonormal sequence $\{\phi_n\}$ is of \emph{the first type} if 
there  exists  $Q$ in \eqref{intr1} such that  $JQ=-QJ$. Otherwise,  $\{\phi_n\}$ is of \emph{the second type}.

We proved that the formula \eqref{intr1} defines first type sequences if and only if $Q$ anticommutes with 
$J$ and elements of ONB $\{e_n\}$ are eigenvectors of $J$.

The first type  sequences have a lot of useful properties.
One of benefits is the fact that a first type sequence  generates a $\cC$-symmetry operator $\cC=e^QJ$
where $Q$ is the same operator as in \eqref{intr1}.  The latter allows one to construct the Hilbert space  $(\cH_{-Q}, <\cdot, \cdot>_{-Q})$
involving $\{\phi_n\}$ as ONB,  directly as the completion of $D(\cC)$ with respect to ``$\cC\mathcal{PT}$-norm'': 
 \begin{equation}\label{AGH19b}
 <\cdot, \cdot>_{-Q}=[\cC\cdot, \cdot]=<Je^QJ\cdot,\cdot>=<e^{-Q}\cdot, \cdot>.
 \end{equation}

For a second type sequence,  the inner product  $<\cdot,\cdot>_{-Q}$  generated by an operator $Q$  from \eqref{intr1} cannot be expressed via  $[\cdot, \cdot]$
 and one should apply more efforts for the determination of $<\cdot,\cdot>_{-Q}$, see Subsection \ref{sec3.2}.

Assume that a complete $J$-orthonormal sequence $\{\phi_n\}$ consists of eigenfunctions of a Hamiltonian $H$
 corresponding to real eigenvalues $\{\lambda_n\}$. The sequence $\{\phi_n\}$ generates $\cC$-symmetry operators  
(Proposition \ref{NEE}). If $\{\phi_n\}$ is of the first type, then there exists at least one $\cC$ such that the operator $H$ restricted on 
 $\mbox{span}\{\phi_n\}$ turns out to be essential self-adjoint in the Hilbert space $\cH_{-Q}$ with the inner product \eqref{AGH19b}.
 The spectrum of the closure of $H$ in $\cH_{-Q}$ coincides with the closure of $\{\lambda_n\}$. Hence,  
 we construct an isospectral realization of $H$ in $\cH_{-Q}$.
 
 For a second type sequence, each operator $\cC$ generated by  $\{\phi_n\}$ gives rise to the Hilbert space  $(\cH_{-Q}, <\cdot, \cdot>_{-Q})$
 with \emph{non-densely defined symmetric operator $H$}. 
 Its extensions to self-adjoint operators in $\cH_{-Q}$ lead to the appearance
of new spectral points. Therefore, self-adjoint realizations constructed with the use of $\cC\mathcal{PT}$-norm cannot be isospectral.
The isospectrality of self-adjoint realizations of $H$ in $\cH_{-Q}$ can be achieved via the construction of the Friedrichs extension $G=e^{-Q}$ of
the symmetric operator $G_0\phi_n=[\phi_n, \phi_n]J\phi_n$ defined on $\mbox{span}\{\phi_n\}$, that
 is quite complicated problem. 

Section \ref{sec5} contains examples of $J$-orthonormal sequences. We show that
the eigenstates of the shifted harmonic oscillator constitute a first type sequence.  This example leads to the following
conjecture: \emph{eigenstates of a $\mathcal{PT}$-symmetric Hamiltonian $H$ with \emph{unbroken} $\mathcal{PT}$-symmetry \cite[p. 41]{BE}
form a first type sequence}.  

In what follows, $\mathcal{H}$ means a complex Hilbert space with inner product  linear in the first argument.
Sometimes, it is useful to specify the inner product $<\cdot,\cdot>$ associated with $\cH$.
In that case the notation $(\cH, <\cdot,\cdot>)$ has already been used, and will be used in the following. 
All operators in $\mathcal{H}$ are supposed to be linear, the identity operator is denoted by $I$.
The symbols ${D}(A)$ and ${R}(A)$ denote the domain and the range of a linear operator
$A$. An operator $A$ is called  positive [nonnegative] if   $(Af,f)>0$  \  [$(Af,f)\geq{0}$] for non-zero $f\in{D}(A)$. 

\section{Generalized Riesz systems in Hilbert spaces}\label{sec2}
Let $\{\phi_n\}$ be a Riesz basis in $\cH$. Then there exists a bounded and boundedly invertible operator
$R$ such that $\phi_n=Re_n'$,  where $\{e_n'\}$ is an orthonormal basis (ONB) of $\cH$.
The  operator $RR^*$ is positive and self-adjoint in $\cH$ and it admits
the presentation $RR^*=e^Q$, where $Q$ is a bounded self-adjoint operator. 
The polar decomposition of $R$ has the form $R=\sqrt{RR^*}U=e^{Q/2}U$, where $U$ is a unitary operator in $\cH$. 
The unitarity of $U$ means that $\{e_n=Ue_n'\}$ is an ONB of $\cH$ and  we can  rewrite the definition of Riesz bases as follows:
 \emph{a sequence $\{\phi_n\}$ is called a Riesz basis if there exists a bounded self-adjoint operator $Q$ in $\cH$ and 
an ONB $\{e_n\}$  such that  $\phi_n=e^{Q/2}e_n$.}
This simple observation leads to: 

\begin{defn}\label{d1}
	A sequence $\{\phi_n\}$ is called a generalized Riesz system (GRS) if there exists a self-adjoint operator $Q$ in $\cH$ and 
	an ONB $\{e_n\}$  such that  $e_n\in{D}(e^{Q/2})\cap{D}(e^{-Q/2})$ and $\phi_n=e^{Q/2}e_n$.
\end{defn}

Let $\{\phi_n\}$ be a GRS. In view of Definition \ref{d1}, the sequence  
$\{\psi_n=e^{-Q/2}e_n\}$ is well defined and it is a biorthogonal sequence for $\{\phi_n=e^{Q/2}e_n\}$.
 Obviously, $\{\psi_n\}$ is a GRS which we will call \emph{a dual GRS}.
 Dual GRS are Riesz bases if and only if $Q$ is a bounded operator. 

{\bf Example 1:--}
 A first simple example of GRS can be extracted from \cite{BB}: if we take $Q=-\frac{x^2}{2}$, $x$ being the position operator, 
 it is clear that $D(e^{Q/2})=\Lc^2(\Bbb R)$, while
$$
D(e^{-Q/2})=D(e^{x^2/4})=\left\{f(x)\in\Lc^2(\Bbb R): \: e^{x^2/4}f(x)\in\Lc^2(\Bbb R) \right\}.
$$
This set is dense in $\Lc^2(\Bbb R)$, since contains each eigenfunction of the quantum harmonic oscillator
\begin{equation}\label{AGH88}
e_n(x)=\frac{1}{\sqrt{2^nn!\sqrt{\pi}}}H_n(x)e^{-x^2/2},  \quad H_n(x)=e^{x^2/2}\left(x-\frac{d}{dx}\right)^ne^{-x^2/2}.
\end{equation}
The Hermite functions $\{e_n(x)\}$ form an orthonormal basis  of $\Lc^2(\mathbb{R})$.

Following \cite{BB}, we have
$\phi_n(x)=e^{Q/2}e_n(x)=\frac{1}{\sqrt{2^n\,n!\sqrt{\pi}}}\,H_n(x)\,e^{-\frac{3x^2}{4}}$, \
$\psi_n(x)=e^{-Q/2}e_n(x)=\frac{1}{\sqrt{2^n\,n!\sqrt{\pi}}}\,H_n(x)\,e^{-\frac{x^2}{4}}$
and  both $\{\phi_n\}$ and $\{\psi_n\}$ are  GRS, one the dual of the other. Of course, due to the unboundedness of $Q$, they are not Riesz bases.
\vspace{3mm}

In general, the inner product 
\begin{equation}\label{new1}
<f, g>_{-Q}:=<e^{-Q}f, g>=<e^{-Q/2}f, e^{-Q/2}g>,  \qquad f,g\in{D}(e^{-Q})
\end{equation}
is not equivalent to $<\cdot,\cdot>$ and the linear space $D(e^{-Q})$ endowed with $<\cdot,\cdot>_{-Q}$  is only a pre-Hilbert
space. Denote by $\cH_{-Q}$ the completion of  $D(e^{-Q})$ with respect to $<\cdot,\cdot>_{-Q}$. 
 Analogously,  the Hilbert space $(\cH_{Q}, <\cdot, \cdot>_Q)$ is defined as the completion of $D(e^Q)$ with respect to the inner product
\begin{equation}\label{new1b}
<f, g>_{Q}:=<e^{Q}f, g>=<e^{Q/2}f, e^{Q/2}g>,  \qquad f,g\in{D}(e^{Q}).
\end{equation}

Generally, the Hilbert spaces
$\cH_{-Q}$ and $\cH_{Q}$ differ from $\cH$ (as the sets of elements). 
We can just say, in view of \eqref{new1} and  \eqref{new1b} that
$D(e^{-Q/2})\subset\cH_{-Q}$ and $D(e^{Q/2})\subset\cH_{Q}$.

\begin{lem}\label{new23}
	Let $\{e_n\}$, $e_n\in{D}(e^{Q/2})\cap{D}(e^{-Q/2})$, be a complete set in $\cH$. 
	Then $\{\phi_n=e^{Q/2}e_n\}$  and  $\{\psi_n=e^{-Q/2}e_n\}$  are complete sets  in the Hilbert spaces $\cH_{-Q}$ and 
	$\cH_{Q}$, respectively.
\end{lem}
\begin{proof} In view of \eqref{new1},  
	\begin{equation}\label{AGH14b}
	\|f\|_{-Q}^2=<e^{-Q}f, f>=\|e^{-Q/2}f\|^2 \quad \mbox{for all}  \quad  f\in{D}(e^{-Q}).
	\end{equation}
Therefore, $\{{f}_n\}$ ($f_n\in{D}(e^{-Q})$) is a Cauchy sequence in $\cH_{-Q}$
	if and only if  $\{e^{-Q/2}{f_n}\}$ is a Cauchy sequence in $\cH$. 
	
	By construction, the vectors $\phi_n$ belong to  $D(e^{-Q/2})$.
	Hence, $\phi_n\in\cH_{-Q}$.  Assume that $f\in\cH_{-Q}$ is orthogonal to $\{\phi_n\}$ and consider
	a sequence $\{f_m\}$  $(f_m\in{D}(e^{-Q}))$ such that ${f}_m\to f\in\cH_{-Q} \  (\mbox{wrt.} \ \|\cdot\|_{-Q})$.
	Then,  $\{e^{-Q/2}{f_m}\}$ is a Cauchy sequence in $\cH$ and  $e^{-Q/2}{f_m}$ converges to some $g\in\cH  \  (\mbox{wrt.} \ \|\cdot\|)$. 
	This means that
	$$
	0=<f, \phi_n>_{-Q}=\lim_{m\to\infty}<{f}_m, e^{Q/2}e_n>_{-Q}=\lim_{m\to\infty}<e^{-Q/2}f_m, e_n>=<g, e_n>
	$$
	and, as a result,  $g=0$ since $\{e_n\}$ is a complete set in $\cH$. This means that  $f=0$ and the set $\{\phi_n\}$ is complete in $\cH_{-Q}$.  Completeness of  
	$\{\psi_n\}$ in $\cH_{Q}$ is established in a similar manner.	
\end{proof}

\begin{prop}\label{new11}
	Every dual GRS $\{\phi_n=e^{Q/2}e_n\}$ and $\{\psi_n=e^{-Q/2}e_n\}$  are ONB of the 
	Hilbert spaces $(\cH_{-Q}, <\cdot, \cdot>_{-Q})$  and  $(\cH_{Q}, <\cdot, \cdot>_{Q})$, respectively.
\end{prop} 
\begin{proof} Due to \eqref{new1},  the sequence $\{\phi_n=e^{Q/2}e_n\}$ is orthonormal in 
	$\cH_{-Q}$. Its completeness follows from Lemma \ref{new23}. 
	The case $\{\psi_n\}$ is considered similarly with the use of \eqref{new1b}.
\end{proof}

 Dual GRS   could be used to define manifestly non self-adjoint Hamiltonians  
\begin{equation}\label{NE1}
	H_{\phi, \psi}f=\sum_{n=1}^\infty \lambda_n\left<f,\psi_n\right>\phi_n,  \qquad
	H_{\psi,\phi}g =\sum_{n=1}^\infty \lambda_n\left<g ,\phi_n\right>\psi_n
\end{equation}
with known complex eigenvalues $\{\lambda_n \}$ and eigenvectors $\{\phi_n\}$ and $\{\psi_n\}$, respectively.
We refer to \cite{BB, BIT} for the connection between $H_{\phi, \psi}$ and the adjoint of $H_{\psi, \phi}$
and for the analysis of ladder operators associated to similar  bi-orthogonal sets, 
and how these ladder operators can be used to factorize the Hamiltonians above.

\subsection{Dual GRS and $\mathcal{G}$-quasi bases}

Dual GRS  $\{\phi_n\}$ and $\{\psi_n\}$ can be considered as examples of more general object: $\mathcal{G}$-quasi bases. 
These are biorthogonal sets originally introduced in \cite{bag2013JMP}, and then analyzed in a series of papers (see \cite{bagbook_thebook} for a relatively recent review). 

\begin{defn}\label{d3} 
	Let $\mathcal{G}$ be a  dense linear manifold  in $\cH$.
	Biorthogonal sequences $\{\phi_n\}$ and $\{\psi_n\}$ are called $\mathcal{G}$-quasi bases,
	if for all   $f, g\in\mathcal{G}$, the following holds:
	\begin{equation}\label{e1}
	<f, g>=\sum_{n=1}^{\infty}<f, \phi_n><\psi_n, g>=\sum_{n=1}^{\infty}<f, \psi_n><\phi_n, g>.
	\end{equation}
\end{defn}

\begin{prop}\label{new25}
	Dual GRS  $\{\phi_n\}$ and $\{\psi_n\}$ are $\mathcal{G}$-quasi bases
	with $\mathcal{G}={D}(e^{Q/2})\cap{D}(e^{-Q/2})$.
\end{prop}
\begin{proof}
	If $\{\phi_n\}$ and $\{\psi_n\}$ are dual GRS, then there exists a self-adjoint operator $Q$ and ONB $\{e_n\}$ such that
	\eqref{intr1} hold.  Hence,  for all 
	$f, g\in\mathcal{G}={D}(e^{Q/2})\cap{D}(e^{-Q/2})$,
	$$
	e^{Q/2}f=\sum<e^{Q/2}f, e_n>e_n=\sum<f, \phi_n>e_n  
	$$
	and
	$$
	e^{-Q/2}g=\sum<e^{-Q/2}g, e_n>e_n=\sum<g, \psi_n>e_n.
	$$
	The last relations yield
	$$
	<f, g>=<e^{Q/2}f, e^{-Q/2}g>=\sum<f,\phi_n><\psi_n, g>.
	$$
	Similarly, $<f, g>=<e^{-Q/2}f, e^{Q/2}g>=\sum<f,\psi_n><\phi_n, g>.$  To complete the proof, it suffices
	notice that $\mathcal{G}$  is dense in $\cH$ since each vector of ONB $\{e_n\}$ belongs to ${D}(e^{Q/2})\cap{D}(e^{-Q/2})$.
\end{proof}

\begin{rem}
	Proposition \ref{new25} implies that Example 1 above of GRS provides  
	also an example of $\mathcal{G}$-quasi bases, with $\mathcal{G}=D(e^{x^2/4})$, in agreement with what was found in \cite{BB}.
\end{rem}

\subsection{Regular biorthogonal sequences and dual GRS}\label{sectIII.3}

 We say that biorthogonal sequences $\{\phi_n\}$ and  $\{\psi_n\}$ are \emph{regular} 
if $\{\phi_n\}$ and  $\{\psi_n\}$ are complete sets in $\cH$.  
In other words,  a biorthogonal sequence $\{\psi_n\}$ is defined uniquely by $\{\phi_n\}$ and vice-versa.

\begin{thm}\label{new31}
	Regular biorthogonal sequences $\{\phi_n\}$ and  $\{\psi_n\}$ are dual GRS.
\end{thm}
\begin{proof}
	Let $\{\phi_n\}$ and $\{\psi_n\}$ be regular biorthogonal sequences. 
	Then an operator $G_0$ defined initially as 
	\begin{equation}\label{new14}
	G_0\phi_n=\psi_n,  \qquad  n\in\mathbb{N}   
	\end{equation}
	and extended on $D(G_0)=span\{\phi_n\}$  by the linearity is 
	densely defined and positive.  The later follows from the fact that
	$$ 
	<G_0f, f>=\sum_{n=1}^k\sum_{m=1}^kc_n\overline{c}_m<\psi_n, \phi_m>=\sum_{n=1}^k{|c_n|^2} 
	$$ 
	for all  $f=\sum_{n=1}^k{c_n}\phi_n\in{D(G_0)}$.
	
	Let $G$ be the Friedrichs extension of $G_0$. The operator $G$ is positive. Indeed, assuming that $Gg=0$ for 
	 $g\in{D}(G)$ we obtain
	$0=<Gg, \phi_n>=<g, G\phi_n>=<g, \psi_n>$.  Therefore, $g=0$ since $\{\psi_n\}$ is a complete set in $\cH$.
	The positivity of $G$ allows one to state that $G=e^{-Q}$,  where $Q$ is a self-adjoint operator in $\cH$.
	Denote
	$e_n=e^{-Q/2}\phi_n$.  Due to \eqref{new14}, $e_n=e^{Q/2}\psi_n$. Therefore,  
	$e_n\in{D}(e^{Q/2})\cap{D}(e^{-Q/2})$    and
	\begin{equation}\label{new15}
	\phi_n=e^{Q/2}e_n,  \qquad  \psi_n=e^{-Q/2}e_n. 
	\end{equation}
	
	The sequence $\{e_n\}$ is orthonormal in $\cH$ since  
	$$
	<e_n, e_m>=<e^{-Q/2}\phi_n,  e^{Q/2}\psi_m>=<\phi_n, \psi_m>=\delta_{nm}.
	$$
	
	Let us assume that $\gamma\in\cH$ is orthogonal to  $\{e_n\}$. 
	Then there exists a sequence $\{f_m\}$  $(f_m\in{D}(e^{-Q}))$ such that  $e^{-Q/2}{f_m}\to\gamma$ in $\cH$
	(because $e^{-Q/2}D(e^{-Q})$ is a dense set in $\cH$).
	In this case, due to \eqref{AGH14b},  $\{f_m\}$ is a Cauchy sequence in $\cH_{-Q}$ and therefore,  ${f}_m$ tends to some $f\in\cH_{-Q}$. 
This means that
    \begin{equation}\label{neww4}\fl
	0=<\gamma,  e_n>=\lim_{m\to\infty}<e^{-Q/2}{f_m}, e_n>=\lim_{m\to\infty}<f_m, \phi_n>_{-Q}=<f, \phi_n>_{-Q}.
	\end{equation}
	We note that the set $D(G_0)=span\{\phi_n\}$ is dense in the Hilbert space $(\cH_{-Q}, <\cdot,\cdot>_{-Q})$  (since $G=e^{-Q}$ is the Friedrichs extension  of 
	$G_0$ \cite{AK_Arlin}). In view of \eqref{neww4}, $f=0$ that means
$\lim_{m\to\infty}\|f_m\|_{-Q}=\lim_{m\to\infty}\|e^{-Q/2}f_m\|=0$ and therefore,  $\gamma=0$.
		 This means that the orthonormal sequence  $\{e_n\}$ is complete in $\cH$. Hence $\{e_n\}$ is an ONB.
\end{proof}

\begin{rem}
Another proof of Theorem \ref{new31} can be found in \cite[Theorem 2.1]{Inoue1}.	
\end{rem}

\subsection{The uniqueness of $Q$ and $\{e_n\}$ for regular biorthogonal sequences}\label{subsection2.2}

Let $\{\phi_n\}$ be a basis in $\cH$.  Then $\{\phi_n\}$ is a regular sequence because
its biorthogonal sequence $\{\psi_n\}$ has to be a basis.  By Theorem \ref{new31}, 
$\{\phi_n\}$ is a GRS, i.e., there exists a self-adjoint operator $Q$ and an ONB $\{e_n\}$ such that  $\phi_n=e^{Q/2}e_n$. 

\begin{prop}\label{lem1}
The operator $Q$ and  ONB $\{e_n\}$ in Definition \ref{d1} are determined \emph{uniquely} for every basis $\{\phi_n\}$.
\end{prop}
\begin{proof}
Let $\gamma$ be orthogonal to $\mathcal{R}(G_0+I)$. Then, in view of \eqref{new14},
$<\gamma, \phi_n>=-<\gamma, \psi_n>$  and the basis property of $\{\phi_n\}$ and $\{\psi_n\}$ leads to the relation:
\begin{equation}\label{KKK3}
\gamma=\sum<\gamma, \psi_n>\phi_n=\sum<\gamma, \phi_n>\psi_n=-\sum<\gamma, \psi_n>\psi_n.
\end{equation}

By virtue of \eqref{KKK3}, the sequence $\gamma_m=\sum_{n=1}^m<\gamma, \psi_n>\phi_n$ tends to $\gamma$, while
$G_0\gamma_m=\sum_{n=1}^m<\gamma, \psi_n>\psi_n$ tends to $-\gamma$. Therefore, 
$$
-\|\gamma\|^2=\lim_{m\to\infty}<G_0\gamma_m, \gamma_m>=\lim_{m\to\infty}\sum_{n=1}^m|<\gamma, \psi_n>|^2=\sum_{n=1}^\infty|<\gamma, \psi_n>|^2
$$ 
that is possible when $\gamma=0$. Hence, $\mathcal{R}(G_0+I)$ is a dense set in $\cH$ and, as a result, $G_0$ is an essentially self-adjoint operator 
in $\cH$. Its closure $\overline{G}_0$ gives   a unique positive self-adjoint extension $G$  which determines a unique self-adjoint operator  $Q=-\ln G$
(i.e. $G=e^{-Q}$). Moreover, because of the equality $e_n=e^{-Q/2}\phi_n$, the ONB $\{e_n\}$  is also determined uniquely. 
\end{proof}

 In view of Proposition \ref{lem1} a natural question arise: \emph{is the operator $Q$ determined uniquely for a  given GRS $\{\phi_n\}$?} 

The choice of the Friedrichs extension $G=e^{-Q}$ of $G_0$ in the proof of Theorem \ref{new31}  was inspired by the fact 
that the sequence $\{\phi_n\}$ must be complete in the Hilbert space $\cH_{-Q}$ 
(that, in view of \eqref{neww4} and Lemma \ref{new23},  is equivalent to the completeness of orthonormal system $\{e_n\}$ in $\cH$). 
Generally, there are many self-adjoint extensions $G$ of $G_0$ which preserve this property and each of them 
can be used instead of the Friedrichs extension. 

We recall \cite{AK_Arlin}  that a nonnegative self-adjoint extension  $G$ of
$G_0$ is called extremal if
$$
\inf_{\phi\in{D}(G_0)}{<G(f-\phi),(f-\phi)>}=0 \quad \mbox{for all} \quad f\in{D}(G).
$$
The Friedrichs extension and the Krein-von Neumann extension of $G_0$ are examples of extremal extensions. 

In the case of regular bi-orthogonal sequences $\{\phi_n\}$ and $\{\psi_n\}$,  the symmetric operator $G_0$ is positive and each nonnegative self-adjoint extension  $G$ of
$G_0$  must also be \emph{positive}. Indeed, if $<Gf,f>=0$ for some $f\in{D}(G)$, then $Gf=0$ and
$$
 0=<Gf, \phi_n>=<f, G\phi_n>=<f, G_0\phi_n>=<f, \psi_n>
$$
that implies $f=0$. Therefore, each nonnegative self-adjoint extension $G$ of $G_0$ is positive and it has the form $G=e^{-Q}$. This means that
 $<G(f-\phi), (f-\phi)>$ coincides with  $\|f-\phi\|^2_{-Q}$ due to  \eqref{AGH14b}.  For this reason, 
 the definition of extremal extensions can be rewritten as follows:  let $G_0$ be determined by \eqref{new14}, 
 where $\{\phi_n\}$ and $\{\psi_n\}$ are  regular biorthogonal sequences. 
A self-adjoint extension $G=e^{-Q}$ of $G_0$ is called \emph{extremal} if 
$$
\inf_{\phi\in{D}(G_0)}\|f-\phi\|^2_{-Q}=0 \quad \mbox{for all} \quad f\in{D}(e^{-Q}).
$$
Therefore, the extremality of a self-adjoint extension $e^{-Q}$ of $G_0$ is equivalent
to the completeness of  the sequence $\{\phi_n\}$ in  $\cH_{-Q}$. This means that for each
extremal self-adjoint extension  $e^{-Q}$ one can repeat  the proof of Theorem \ref{new31}
 and establish the relations  \eqref{new15}, where $\{e_n\}$ is an ONB. 
Summing up, we prove the equivalence of statements $(i)$ and $(ii)$  in:
\begin{prop}\label{KKK5}
Let $\{\phi_n\}$ and $\{\psi_n\}$ be a regular biorthogonal sequences. The following are equivalent:
\begin{enumerate}
\item[(i)] the self-adjoint operator $Q$ and the ONB  $\{e_n\}$  are determined uniquely in 
\eqref{new15};
\item[(ii)] the symmetric operator $G_0$ in \eqref{new14} has a unique extremal extension $G=e^{-Q}$.
\item[(iii)] 
\begin{equation}\label{KKK6}
\inf_{\phi\in{D}(G_0)}\frac{<G_0\phi,\phi>}{|<\phi,g>|^2}=0,  \quad \mbox{for all nonzero} \quad g\in\ker(I+G_0^*),
\end{equation}
where $G_0^*$ means the adjoint operator of $G_0$  with respect to $<\cdot,\cdot>$.
\end{enumerate}
\end{prop}
\begin{proof} It suffices to establish the equivalence $(ii)$ and $(iii)$. 
Indeed, the set of extremal extensions involves the Friedrichs $G_F$ and the Krein-von Neumann $G_K$ extensions of $G_0$  and 
it contains only one element when $G_F=G_K$ \cite{AK_Arlin}. The last equality is equivalent to  \eqref{KKK6} due to \cite[Theorem 9]{AK_Krein}. 
 \end{proof}

\section{Orthonormal sequences in Krein space}\label{sec4}	

\subsection{Elements of the Krein spaces theory}\label{sec4a}

Here all necessary results of Krein spaces theory are presented in a form convenient for our exposition.
The chapters \cite[Chap. 6]{bagbook_thebook} and \cite[Chap. 8]{BE}
are recommended as complementary reading on the subject.

An operator $J$ is called \emph{fundamental symmetry} in a Hilbert space $\cH$ if
$J$ is a bounded self-adjoint operator in $\cH$  and  $J^2=I$.

Let $J$ be a non-trivial fundamental symmetry, i.e., $J\not={\pm{I}}$. The Hilbert space
$(\cH, <\cdot, \cdot>)$ equipped with the indefinite inner product
$[\cdot, \cdot]:=<J\cdot, \cdot>$  
is called a Krein space $(\cH, [\cdot,\cdot])$.

The principal difference between the initial inner product $<\cdot,\cdot>$ and the indefinite inner product
$[\cdot,\cdot]$ is that there exist nonzero elements $f\in\cH$ such that 
$[f,f]<0$. An element $f\not=0$ is called  \emph{positive} or \emph{negative} if
 $[f,f]>0$ or $[f,f]<0$, respectively.
A  closed subspace $\mathfrak{L}$ of the Hilbert space $(\cH, <\cdot,\cdot>)$  is called \emph{positive} or {\em negative} 
if all nonzero elements $f\in\mathfrak{L}$ are, respectively, positive or negative.
A positive (negative) subspace $\mathfrak{L}$ is called
\emph{uniformly positive} (\emph{uniformly negative}) if there exists $\alpha>0$ such that
$$
[f,f]\geq\alpha<f,f> \qquad  (-[f,f]\geq\alpha<f,f>) \quad \forall{f}\in\mathfrak{L}.
$$

In each of these classes we can define maximal subspaces.
For instance, a positive subspace $\mathfrak{L}$ is called \emph{maximal positive} if $\mathfrak{L}$
is not a subspace of another positive subspace in $\cH$.  
The maximality of a (negative, uniformly positive, uniformly negative) closed subspace
is defined similarly.

Let a subspace ${\mathfrak L}$ be a maximal positive (negative). Then its
orthogonal complement with respect to the indefinite inner product $[\cdot, \cdot]$ 
$$
{\mathfrak L}^{[\bot]}=\{f\in\cH \ :\ [f,g]=0,\ \forall{g}\in
{\mathfrak L}\}
$$
is a maximal negative (positive) subspace and the $J$-orthogonal  sum
\begin{equation}\label{e8}
{\mathfrak L}[\dot{+}]{\mathfrak L}^{[\bot]}
\end{equation}
is dense in the Hilbert space $(\cH, <\cdot, \cdot>)$
(the symbol $[\dot{+}]$ in \eqref{e8} indicates that the subspaces ${\mathfrak L}$ and
${\mathfrak L}^{[\bot]}$  are orthogonal with respect to $[\cdot,\cdot]$, i.e. $J$-orthogonal).

The $J$-orthogonal  sum \eqref{e8} coincides with $\cH$, i.e.,
\begin{equation}\label{e9}
\cH={\mathfrak L}[\dot{+}]{\mathfrak L}^{[\bot]}
\end{equation}
 if and only if ${\mathfrak L}$ is  a maximal uniformly positive (uniformly negative) subspace 
(in this case,  ${\mathfrak L}^{[\bot]}$ is maximal uniformly negative (uniformly positive)).

\begin{rem}
 The decomposition \eqref{e9} is called \emph{the fundamental decomposition} of $\cH$ and it 
is often used for (an equivalent) definition of Krein spaces. Precisely, let $\cH$  be a complex linear space with a
Hermitian sesquilinear form $[\cdot,\cdot]$  
(i.e. a mapping $[\cdot,\cdot]:{\cH}\times{\cH}\to\mathbb{C}$ such that  $
[\alpha_1f_1+\alpha_2f_2,g]=\alpha_1[f_1,g]+\alpha_2[f_2,g]$ and 
$[f,g]=\overline{[g,f]}$ for all $f_1, f_2, f, g\in{\cH}$, $\alpha_1, \alpha_2\in\mathbb{C}$). 
Then $({\cH}, [\cdot,\cdot])$ is called a Krein space if $\cH$ admits a decomposition
(\ref{e9}) such that the linear manifolds $({\mathfrak L}, [\cdot, \cdot])$ and
 $({\mathfrak L}^{[\bot]}, -[\cdot, \cdot])$ are Hilbert spaces (here we suppose for definiteness that
 ${\mathfrak L}$ is positive).
\end{rem}

 Each fundamental decomposition \eqref{e9} is uniquely determined by a bounded operator 
 $\cC$ which coincides with the identity operator on the positive subspace ${\mathfrak L}_+:={\mathfrak L}$
and with the minus identity operator on the negative subspace ${\mathfrak L}_-:={\mathfrak L}^{[\bot]}$. 
By the construction, 
${\mathfrak L}_\pm=(I\pm\cC)\cH$  and  $\cC^2=I$. Moreover,
the operator $J\cC$ is positive self-adjoint since
$$
<J\cC{f}, f>=[\cC{f}, f]=[f_+, f_+]-[f_-, f_-]>0 \quad \mbox{for non-zero} \quad f=f_++f_-, \ f_\pm\in{\mathfrak L}_\pm.
$$
Hence, $J\cC=e^{-Q}$, where $Q$ is a bounded self-adjoint operator. The relations $\cC^2=I, J\cC=e^{-Q}>0$ 
 and  \cite[Theorem 2.1]{KS} imply that 
\begin{equation}\label{new34}
	JQ=-QJ.
\end{equation}

Similar reasonings applied to the $J$-orthogonal sum \eqref{e8} gives rise to the collection of unbounded operators
$\cC=Je^{-Q}=e^Q{J}$, where unbounded $Q$ anticommutes with $J$. The subspaces in \eqref{e8}
are recovered as ${\mathfrak L}_\pm=(I\pm\cC)({\mathfrak L}_+[\dot{+}]{\mathfrak L}_-)$. 

Summing up: \emph{the fundamental decompositions  \eqref{e9} of a Krein space are in one-to-one correspondence
with the set of bounded operators $\cC=Je^{-Q}=e^Q{J}$.}

\emph{The $J$-orthogonal sums \eqref{e8} of maximal positive/maximal negative subspaces are in one-to-one correspondence
with the set of unbounded operators $\cC=Je^{-Q}=e^Q{J}$.  In both cases,  $Q$ anticommutes with $J$.}

The operator $\cC$ is called \emph{a $\cC$-symmetry operator} and this notion is widely used 
in ${\mathcal P}{\mathcal T}$-symmetric approach in Quantum Mechanics \cite{BE}.

\begin{rem} If  $Q$ in \eqref{new34} is unbounded, then  we understood  \eqref{new34} 
as the identity $JQf=-QJf$, where $f\in{D(Q)}$ and
$J$ leaves $D(Q)$ invariant. From now on, we will adopt this simplifying notation.
\end{rem}

A $\cC$-symmetry operator allows one to define a new inner product via the indefinite inner product $[\cdot, \cdot]$:
\begin{equation}\label{AGHNEW}\fl
<f, g>_{-Q}:=[\cC{f}, g]=<J^2e^{-Q}f,g>=<e^{-Q}f,g>, \qquad f,g\in{D}(\cC)=D(e^{-Q}).
\end{equation}
The corresponding norm $\|\cdot\|_{-Q}$ is equivalent to the original norm of $\cH$ when $\cC$  is bounded.
If $\cC$ is unbounded, then the completion of ${D}(\cC)$ with respect to $\|\cdot\|_{-Q}$
leads to the Hilbert space $(\cH_{-Q}, <\cdot, \cdot>_{-Q})$ defined in Section \ref{sec2}.

\subsection{$J$-orthonormal sequences of the first and of the second type}\label{sec3.2}	 
A sequence $\{\phi_n\}$ is called orthonormal in a Krein space $(\cH, [\cdot,\cdot])$ (briefly,  $J$-orthonormal)
if  $|[\phi_n, \phi_m]|=\delta_{nm}$.  For each  $J$-orthonormal sequence $\{\phi_n\}$ there exists a biorthogonal one
 \begin{equation}\label{new4}
	\psi_n=[\phi_n,\phi_n]J\phi_n.
	\end{equation}
Obviously, $\{\psi_n\}$ is $J$-orthonormal and $[\phi_n, \phi_n]=[\psi_n, \psi_n]$.
In view of \eqref{new4}, the positive symmetric operator $G_0$ in \eqref{new14} acts as
\begin{equation}\label{GGG}
G_0\phi_n=[\phi_n,\phi_n]J\phi_n,  \qquad  n\in\mathbb{N}.
\end{equation}
 
 In what follows we assume that $\{\phi_n\}$ is complete in the Hilbert space $\cH$. 
 Then $\{\psi_n\}$ in \eqref{new4} is complete too. 
Therefore, $\{\phi_n\}$  and  $\{\psi_n\}$ are regular biorthogonal sequences and,
by Theorem \ref{new31},  they are dual GRS. Thus, \emph{each complete $J$-orthonormal sequence is a GRS.}
The corresponding operator  $Q=-\ln \ G$ in \eqref{new15} can be determined by every extremal extension $G=e^{-Q}$ of $G_0$. 
Such kind of freedom allows us to select an appropriative operator $Q$  which fits well with
 the $J$-orthonormality of $\{\phi_n\}$. 
 
\begin{thm}\label{new2}
	Let  $\{\phi_n\}$  be a complete $J$-orthonormal sequence. If a  self-adjoint operator $Q$ 
	in \eqref{new15} is determined uniquely, then the relation \eqref{new34} holds.
\end{thm}
\begin{proof} 
Separating the sequence  $\{\phi_n\}$ by the signs of $[\phi_n,\phi_n]$:
\begin{equation}\label{bebe95}
\phi_{n}=\left\{\begin{array}{l}
\phi_{n}^+ \quad \mbox{if} \quad [\phi_{n},\phi_{n}]=1, \\
\phi_{n}^- \quad \mbox{if} \quad [\phi_{n}, \phi_{n}]=-1
\end{array}\right.
\end{equation}
we obtain two sequences of positive $\{\phi_n^+\}$ and negative $\{\phi_n^-\}$  elements.
Denote by  ${\mathfrak L}_+^0$ and ${\mathfrak L}_-^0$  the closure (in the Hilbert space $\cH$) of
the linear spans generated by the sets $\{\phi_n^+\}$ and $\{\phi_n^-\}$, respectively. 
By construction, ${\mathfrak L}_\pm^0$ are  positive/negative subspaces and their  $J$-orthogonal sum 
${\mathfrak L}_+^0[\dot{+}]{\mathfrak L}_-^0$ coincides with the domain of the closure $\overline{G}_0$ of $G_0$ determined
by \eqref{GGG} on $span\{\phi_n\}$ \cite[Lemma 4.1]{KS}\footnote{in \cite{KS}, the notation $G_0$ is used for $\overline{G}_0$}.  

By Proposition \ref{KKK5}, the uniqueness of $Q$  means that
the symmetric operator $\overline{G}_0$ has a unique extremal extension $G=e^{-Q}$.
This is possible when $G$ coincides with the Friedrichs extension of $G_0$ as well as with the Krein-von Neumann extension of $G_0$.  
This fact, by virtue of \cite[Theorem 4.3]{KS}, means that 
$Je^{-Q}f=e^QJf$  for  $f\in{D}(e^{-Q}).$  The last relation and  \cite[Theorem 2.1]{KS}  justify \eqref{new34}.  
\end{proof}
 
In view of Theorem \ref{new2},  it seems natural  to consider the anti-commutation relation \eqref{new34} in the  case
where $Q$ is not determined uniquely. Taking into account that \eqref{new34} is equivalent 
to the relation 
\begin{equation}\label{KKK9}
JGf=G^{-1}Jf, \qquad f\in{D}(G), 
\end{equation}
where $G=e^{-Q}$ is an extremal extension of $G_0$, we reduce the choice of $Q$ which satisfies  \eqref{new34}  to the choice
of an extremal extension $G$ satisfying \eqref{KKK9}. 

If extremal extensions $G$ of $G_0$ are not determined uniquely, then
 not each $Q=-\ln{G}$ will  anticommute necessarily with $J$.
 In particular, the operator $Q$ that corresponds to the Friedrichs extension $G=e^{-Q}$ of $G_0$ does not satisfy \eqref{new34}
  \cite{KKS}.   

\begin{defn}\label{newdef}
A complete $J$-orthonormal sequence $\{\phi_n\}$ is of the first type if 
there exists a  self-adjoint operator $Q$ such that the formulas
 \eqref{new15} hold with the additional property $JQ=-QJ$.
 Otherwise,  $\{\phi_n\}$ is of the second type.
\end{defn}

In view of Theorem \ref{new2}, each complete $J$-orthonormal 
sequence $\{\phi_n\}$ with the unique operator $Q$ in \eqref{new15} is the first type.
In particular, every $J$-orthonormal basis is a first type sequence.  The example of 
a second type sequence can be found in  \cite[Subsection 6.2]{KKS}.
In what follows, considering a first type sequence, we assume that $Q$  anti-commutes with $J$.

\begin{prop}\label{new2b}
Let a complete sequence  $\{\phi_n\}$ be a GRS. The following are equivalent: 
\begin{enumerate}
\item[(i)]  $\{\phi_n\}$ is the first type; 
\item[(ii)]   the operator $Q$ in \eqref{new15} can be chosen in such a way that $JQ=-QJ$ and the vectors $e_n$ are eigenvectors of $J$
(i.e., $Je_n=e_n$ or $Je_n=-e_n$).
\end{enumerate}
\end{prop}
\begin{proof} 
$(i)\to(ii)$.  
By virtue of \eqref{new15} and \eqref{new4},
$$
J\phi_n=Je^{Q/2}e_n=e^{-Q/2}Je_n=[\phi_n, \phi_n]\psi_n=[\phi_n, \phi_n]e^{-Q/2}e_n.
$$
Comparing the third and the fifth terms in the equality above we get $Je_n=[\phi_n, \phi_n]e_n$ that implies $(ii)$.

$(ii)\to(i)$.  Since $\{\phi_n\}$ is complete in $\Hil$ by assumption, it suffices  to verify the $J$-orthonormality of $\{\phi_n\}$:  
$[\phi_n, \phi_m]=<Je^{Q/2}e_n, e^{Q/2}e_m>=<e^{-Q/2}Je_n, e^{Q/2}e_m>=[e_n, e_n]\delta_{nm}.$
\end{proof}

\begin{rem}
The studies of the first type sequences began in \cite{KKS}, where they were called ``quasi bases''. 
Proposition \ref{new2b} is a part of \cite[Theorem 6.3]{KKS}. We present here a simpler proof.
\end{rem}

For the first type sequence,  the  inner product in  $(\cH_{-Q}, <\cdot, \cdot>_{-Q})$ 
is \emph{directly determined by the known indefinite inner product}
$[\cdot, \cdot]$, see \eqref{AGH25} below.
 Let us briefly explain this important fact  (see \cite{KKS} for details).

Since  $Q$  anticommutes with $J$, the $J$-orthogonal sum
${\mathfrak L}_+^0[\dot{+}]{\mathfrak L}_-^0$ of the subspaces ${\mathfrak L}_\pm^0$ defined in the proof
of Theorem \ref{new2} can be extended to the $J$-orthogonal sum
$$
D(G)=D(e^{-Q})={\mathfrak L}_+[\dot{+}]{\mathfrak L}_-, \qquad  {\mathfrak L}_\pm^0\subset{\mathfrak L}_\pm,  
$$
where ${\mathfrak L}_\pm$ are maximal positive/negative subspaces in the Krein space $(\cH, [\cdot,\cdot])$ and 
they are uniquely determined by the choice of $Q$:  ${\mathfrak L}_\pm=(I\pm{Je^{-Q}})D(e^{-Q})$.  The last relation
and \eqref{new1} imply that for $f=(I+Je^{-Q})u$ and  $g=(I+Je^{-Q})v$ from ${\mathfrak L}_+$:
\begin{eqnarray*}\fl
[f, g]=<Jf, g>=<J(I+Je^{-Q})u, (I+Je^{-Q})v>=2([u, v]+<e^{-Q}u, v>) & = & \\
<e^{-Q}(I+Je^{-Q})u, (I+Je^{-Q})v>=<e^{-Q}f, g>=<f,g>_{-Q}.
\end{eqnarray*}
Therefore, the indefinite inner product $[\cdot,\cdot]$ 
coincides with $<\cdot,\cdot>_{-Q}$ on ${\mathfrak L}_+$. Similar calculations show that $[\cdot,\cdot]$ coincides with   
$-<\cdot,\cdot>_{-Q}$ on ${\mathfrak L}_-$.  Moreover,  the subspaces ${\mathfrak L}_\pm$ are orthogonal with respect 
to $<\cdot,\cdot>_{-Q}$ since
$$
<f, \gamma>_{-Q}=<e^{-Q}(I+Je^{-Q})u, (I-Je^{-Q})w>=0, 
$$
where  $\gamma=(I-Je^{-Q})w$ and  $u, w\in{D(e^{-Q})}$. This leads to the conclusion that
\begin{equation}\label{AGH14}
\cH_{-Q}=\widehat{{\mathfrak L}}_+[\oplus_{-Q}]\widehat{\mathfrak L}_-,
\end{equation}
where $\widehat{{\mathfrak L}}_\pm$ are the completion of the pre-Hilbert spaces $({{\mathfrak L}}_\pm, \pm[\cdot,\cdot])$
and $[\oplus_{-Q}]$ indicates the orthogonality with respect to $<\cdot,\cdot>_{-Q}$ and with respect to $[\cdot, \cdot]$. 
Keeping the same notation  for the extension of $[\cdot,\cdot]$ onto $\cH_{-Q}$ we obtain
the new Krein space $(\cH_{-Q}, [\cdot,\cdot])$ with the  fundamental decomposition \eqref{AGH14}.
For every $f,g\in\cH_{-Q}$ \ ($f_\pm, g_\pm\in\widehat{{\mathfrak L}}_\pm$),
\begin{equation}\label{AGH25}  
<f, g>_{-Q}=[f_+,g_+] - [f_-,g_-].    
\end{equation}

For the second type sequences,  there are no operators $Q$ in \eqref{new15} 
which anticommute with $J$. The space $\cH_{-Q}$ cannot be presented as in \eqref{AGH14}. 
This implies that $<\cdot,\cdot>_{-Q}$ cannot be directly expressed via  $[\cdot,\cdot]$
and one should apply  much more efforts for calculation of $<\cdot,\cdot>_{-Q}$.

Let $\{\phi_n\}$ be the first type.  Then the linear manifold $\mathcal{G}={D}(e^{Q/2})\cap{D}(e^{-Q/2})$  in 
Proposition \ref{new25} is invariant with respect to $J$ and the formula \eqref{e1} can be rewritten as
$$
[f, g]=\sum_{n=1}^{\infty}\delta_n[f, \phi_n][\phi_n, g]=\sum_{n=1}^{\infty}\delta_n[f, \psi_n][\psi_n, g],  \qquad   f,g\in\mathcal{G},
$$
where $\delta_n=[\phi_n,\phi_n]=[\psi_n,\psi_n]$.
Moreover, for all $f\in\cH_{-Q}$,
\begin{equation}\label{AGH41}
f=\sum_{n=1}^\infty\delta_n[f, \phi_n]\phi_n,
\end{equation}
where the series converges in $(\cH_{-Q}, <\cdot,\cdot>_{-Q})$.  Indeed, 
$f=\sum_{n=1}^\infty<f, \phi_n>_{-Q}\phi_n$ since
 $\{\phi_n\}$ is  ONB of $\cH_{-Q}$ (Proposition \ref{new11}). By virtue of \eqref{bebe95}, \eqref{AGH14} and \eqref{AGH25},
$$
<f, \phi_n>_{-Q}=\left\{\begin{array}{c}
[f, \phi_n^+] \quad (\mbox{if} \ \phi_n=\phi_n^+); \\
 -[f, \phi_n^-]  \quad (\mbox{if} \ \phi_n=\phi_n^-)
 \end{array}\right. =\delta_n[f, \phi_n]
$$
 that implies \eqref{AGH41}. 

Assume that  $\sum_{n=1}^\infty{c_n}\phi_n$  converges to an element $f$ in $(\cH, <\cdot,\cdot>)$. 
In this case, due to \eqref{new4}, $c_n=<f,\psi_n>=\delta_n[f,\phi_n].$
In general, we cannot state that this series  converges  unconditionally in $\cH$.
Denote
$$
\mathcal{D}_{un}=\{f\in\cH :  \ \mbox{the series} \ \sum_{n=1}^\infty\delta_n[f,\phi_n]\phi_n  \ \mbox{converges unconditionally to}  \ f \ \mbox{in} \ \cH\}.
$$

\begin{prop}\label{fr75}
Let $\{\phi_n\}$ be a first type sequence. Then 
$\mathcal{D}_{un}\subseteq\mathfrak{L}_{+}^0\dot{+}\mathfrak{L}_{-}^0$,
where $\mathfrak{L}_{\pm}^0$ are defined in the proof of Theorem \ref{new2}.
\end{prop}
\begin{proof}
Let $f\in\mathcal{D}_{un}$. Then, simultaneously with \eqref{AGH41}, the series 
$$
\sum_{n}^\infty[f, \phi_n^+]\phi_n^+ \qquad  \mbox{and} \qquad  -\sum_{n}^\infty[f, \phi_n^-]\phi_n^- 
$$
(the vectors $\phi_n^\pm$ are determined by \eqref{bebe95}) converge to elements $f_\pm$ in the Hilbert space
$\cH$ (see, e.g., \cite[Theorem 3.10]{Heil}).  By the construction $f_\pm\in\mathfrak{L}_\pm^0$. Therefore,
$f=f_++f_-$ belongs to  $\mathfrak{L}_{+}^0\dot{+}\mathfrak{L}_{-}^0$.
\end{proof}   

\subsection{$J$-orthonormal sequences and operators of $\cC$-symmetry}\label{sec3.3}

We say that an $J$-orthonormal sequence $\{\phi_n\}$  \emph{generates a  $\cC$-symmetry operator
  $\cC=Je^{-Q}=e^{Q}J$} (the operator $Q$ anti-commutes with $J$) if  
$$
\cC\phi_n^+=\phi_n^+, \qquad  \cC\phi_n^-=-\phi_n^-,
$$
where $\phi_n^\pm$ are defined in \eqref{bebe95}.

With each operator $\cC$ one can associate a Hilbert space  $(\cH_{-Q}, <\cdot, \cdot>_{-Q})$ (see Subsection \ref{sec4a}).
In view of \eqref{AGHNEW},
$$
<\phi_n, \phi_m>_{-Q}=[\cC{\phi_n}, \phi_m]=\left\{\begin{array}{c}
[\phi_n^+, \phi_m] \quad (\mbox{if} \ \phi_n=\phi_n^+) \\
 -[\phi_n^-, \phi_m]  \quad (\mbox{if} \ \phi_n=\phi_n^-)
 \end{array}\right. =\delta_{nm}.
$$
Therefore, $\{\phi_n\}$ is an orthonormal system in  $\cH_{-Q}$.

Proposition \ref{NEE} and Corollary \ref{NEE1} follow from \cite[Sections 5, 6]{KKS}.
\begin{prop}\label{NEE}
Each complete $J$-orthonormal sequence $\{\phi_n\}$ generates at least one operator of $\cC$-symmetry. 
The sequence $\{\phi_n\}$ is the first type if and only if it generates  a $\cC$-symmetry operator  $\cC=e^{Q}J$
such that $\{\phi_n\}$ is an ONB of $\cH_{-Q}$. This operator $\cC$ is determined uniquely or 
there are infinitely many such operators.  
\end{prop}

Obviously, if $\{\phi_n\}$ is the first type, then its bi-orthogonal sequence $\{\psi_n\}$
is also the first type. 

 \begin{cor}\label{NEE1}
If $\{\phi_n\}$ is the first type and $\{\lambda_n\}$ are real numbers, then
there exists a $\cC$-symmetry operator  $\cC=e^{Q}J$ such that the operators 
$H_{\phi,\psi}$ and $H_{\psi,\phi}$ defined in \eqref{NE1} on the domains
 $D(H_{\phi,\psi})=\mbox{span}\{\phi_n\}$ and $D(H_{\psi,\phi})=\mbox{span}\{\psi_n\}$, respectively are essentially self-adjoint in the Hilbert spaces 
$\cH_{-Q}$ and $\cH_{Q}$.  The spectra of  $H_{\phi,\psi}$ and $H_{\psi,\phi}$ coincides with the closure of $\{\lambda_n\}$.
\end{cor}
 
\begin{rem}
Corollary \ref{NEE1} can be easy extended for the general case of GRS $\{\phi_n=e^{Q/2}e_n\}$. 
Indeed, in view of Proposition \ref{new11}, $\{\phi_n\}$  and $\{\psi_n\}$  are ONB of the Hilbert spaces $\cH_{-Q}$ and $\cH_{Q}$, respectively. 
Therefore, the operators  $H_{\phi,\psi}$ and $H_{\psi,\phi}$ defined on $\mbox{span}\{\phi_n\}$ and $\mbox{span}\{\psi_n\}$
have to be essentially self-adjoint in $\cH_{-Q}$ and $\cH_{Q}$. 
 This approach can be used for $J$-orthonormal sequences of the second type. 
The principal difference is:
for the first type sequence, the new scalar product in  $\cH_{-Q}$ is  directly determined by the known  indefinite inner product
$[\cdot, \cdot]$, see \eqref{AGH25}. 
For the second type sequence,  the inner product  $<\cdot,\cdot>_{-Q}$ cannot be expressed via  $[\cdot, \cdot]$
 and it becomes more complicated to determine $<\cdot,\cdot>_{-Q}$.
\end{rem}
 
\section{Examples}\label{sec5}

\subsection{Eigenfunctions of the shifted harmonic oscillator}
 
 In the space $\Lc^2(\mathbb{R})$ we define a fundamental symmetry $J$ as  the space parity operator $\mathcal{P}f(x)=f(-x)$.
 The  indefinite inner product is
 $$
 [f, g]= <\mathcal{P}f, g>=\int_{-\infty}^\infty{f(-x)}\overline{g(x)}dx.
 $$
  The Hermite functions $e_n(x)$ in \eqref{AGH88} form an ONB of 
  $\Lc^2(\mathbb{R})$ and  $\mathcal{P}e_n=(-1)^ne_n$.
 
 Define  
\begin{equation}\label{AGH100}
 \phi_n(x)=e_n(x+ia),  \quad \psi_n(x)=e_n(x-ia),  \quad    a\in\mathbb{R}\setminus\{0\}, \quad n=0,1,2,\ldots
\end{equation}
using the fact that $e_n(x)$ are entire functions. The functions $\{\phi_n\}$ and $\{\psi_n\}$ are
eigenvectors of the shifted harmonic  oscillators  
$$
H=-\frac{d^2}{dx^2}+x^2+2iax  \quad  \mbox{and} \quad H^*=-\frac{d^2}{dx^2}+x^2-2iax,
$$
respectively.  It follows from \cite[Lemma 2.5]{Mit}  that $\{\phi_n\}$ and $\{\psi_n\}$ are regular biorthogonal
sequences and hence,  they are dual GRS (Theorem \ref{new31}).  

To find a self-adjoint operator $Q$ in \eqref{new15}, we calculate the Fourier transform
$F\phi_n=\frac{1}{\sqrt{2\pi}}\int_{-\infty}^\infty{e^{-ix\xi}}\phi_n(x)dx$.
In view of \eqref{AGH100},  $F\phi_n=e^{-a\xi}Fe_n$. 
Therefore,  $\phi_n=F^{-1}e^{-a\xi}Fe_n$.
The last relation can be rewritten as
\begin{equation}\label{AGH101}
\phi_n=e^{Q/2}e_n,  \qquad  e^{Q/2}=F^{-1}e^{-a\xi}F,
\end{equation}
where $Q=2ai\frac{d}{dx}$ is an unbounded self-adjoint operator in $\Lc^2(\mathbb{R})$ that anticommutes with $\mathcal{P}$.
By virtue of Proposition \ref{new2b},   $\{\phi_n\}$  and $\{\psi_n\}$ are $\mathcal{P}$-orhonormal sequences of the first type.
The operator of $\cC$-symmetry generated $\{\phi_n\}$ coincides with $\cC=e^{ai\frac{d}{dx}}\mathcal{P}$ and it is determined 
uniquely.
  
  \subsection{Eigenfunctions of perturbed anharmonic oscillator}

Let $\{e_n\}$ be eigenfunctions  of the anharmonic oscillator 
$$
H_\beta=-\frac{d^2}{dx^2} + |x|^\beta, \qquad  \beta>2
$$ 
Without loss of generality we assume that $\|e_n\|=1$. 
The eigenfunctions $\{e_n\}$ form an orthonormal basis in $\Lc^2(\mathbb{R})$
and they are eigenvectors of $\mathcal{P}$. 

Consider the sequences
 $$
 \phi_n(x)=e^{p(x)}e_n(x),  \qquad  \psi_n(x)=e^{-p(x)}e_n(x) 
 $$
where a real-valued odd function $p\in{C^2}(\mathbb{R})$ 
satisfies \cite[Assumption II]{Mit}. The functions 
$\{\phi_n\}$ and $\{\psi_n\}$ are
eigenvectors of the perturbed anharmonic  oscillators  
$$
H=H_\beta+2p'(x)\frac{d}{dx}+p''(x)- (p'(x))^2  \quad  \mbox{and} \quad H^*=H_\beta-2p'(x)\frac{d}{dx}-p''(x)- (p'(x))^2,
$$
respectively.  The functions  $\{\phi_n\}$ and $\{\psi_n\}$
are determined by \eqref{new15} with the operator of multiplication $Q=2p(x)$ in 
$\Lc^2(\mathbb{R})$ which anticommutes with ${\mathcal P}$.
 Proposition \ref{new2b} implies that  $\{\phi_n\}$ and $\{\psi_n\}$ are $\mathcal{P}$-orhonormal and they are sequences of the first type. 

\subsection{Back to the harmonic oscillator}

Let us go back to Example 1 in Section \ref{sec2}. The vectors $\phi_n(x)$ and $\psi_n(x)$ are respectively eigenstates of $H$ and $H^*$, where
$$
H=\frac{1}{2}\left(-\frac{d^2}{dx^2}\,-x\frac{d}{dx}+\frac{1}{2}\left(\frac{3x^2}{2}-1\right)\right):
$$
$H\phi_n=E_n\phi_n$ and $H^*\psi_n=E_n\psi_n$, where $E_n=n+\frac{1}{2}$. Simple computations show that $H^*=\frac{1}{2}\left(-\frac{d^2}{dx^2}\,+x\frac{d}{dx}+\frac{1}{2}\left(\frac{3x^2}{2}+1\right)\right)$, see \cite{BB}. The sets $\{\phi_n\}$ and $\{\psi_n\}$ are both complete in $\Lc^2(\Bbb R)$. 
Hence they are regular and dual GRS, as a consequence of Theorem \ref{new31}. 

The interesting aspect of this example is that the operator $Q=-\frac{x^2}{2}$ does not anticommute with
the space parity operator $\mathcal{P}$. In fact, we have $Q\mathcal{P}=\mathcal{P}Q$. 
Then Proposition \ref{new2b} suggests that the set $\{\phi_n\}$ cannot be $\mathcal{P}$-orthonormal. 
In fact, it is possible to check that this is what happens. To do that, we first notice that, because of the parity properties of the Hermite polynomials, we have
$$
[\phi_n,\phi_m]=\frac{(-1)^n}{\sqrt{2^{n+m}\,n!\,m!\,\sqrt{2}}}\,\int_{\Bbb R}H_n(x)\,H_m(x)e^{-\frac{3x^2}{2}}\,dx.
$$
This integral is zero if $n+m$ is odd. But, if $n+m$ is even, the result is non zero. In fact, after some computations, we get
$$
|[\phi_n,\phi_m]|=\sqrt{\frac{2^{n+m+1}}{3^{n+m+1}\,\pi\,n!\,m!}}\,\Gamma\left(\frac{n+m+1}{2}\right) {}_2F_1\left(-m,n;\frac{1-m-n}{2};\frac{3}{2}\right),
$$
where $\Gamma$ and ${}_2F_1$ are respectively the Gamma and the Hypergeometric functions. The result is not zero, if $n+m$ is even.
 Hence the set  $\{\phi_n\}$ is not $\mathcal{P}$-orthonormal, as expected. 

\section{Conclusions}
In this paper, motivated by many (already existing or) possible applications to 
$\mathcal{PT}$-quantum mechanics, we have derived some properties of vectors which are complete in a given Hilbert space, 
and orthonormal with respect to an indefinite inner product. 
In this analysis we have heavily used generalized Riesz systems and $\mathcal{G}$-quasi bases, 
and we have introduced two different types of $J$-orthonormal sequences, those of the first and those of the second type, 
depending on the validity of a certain anti-commutation relation between $J$ and $Q$. 
Examples of both kind have been proposed, all arising from harmonic or anharmonic oscillators.

Among our future projects we plan to study physical operators constructed, following 
\cite{BIT, Inoue1}, by the bi-orthogonal sets considered  along this paper and to check under which conditions a given Hamiltonian can be factorized. 
When this is possible, we will also consider the properties of its SUSY partner.

\section*{Acknowledgements}
FB   acknowledges support from the GNFM of Indam and from the University of Palermo, via CORI.
SK  acknowledges support from the Polish Ministry of Science and
Higher Education.

\end{document}